\documentclass{article}
\usepackage{amsmath,amsxtra,amssymb,amsthm,amsfonts}
\usepackage{graphicx, tikz} 
\usepackage{color,graphicx}
\usepackage[margin=1.05in]{geometry}
\usepackage{array}
\usepackage[T1]{fontenc}
\usepackage{tikz}
\usepackage{hyperref}
\usepackage{cleveref}

\numberwithin{equation}{section}
\newtheorem{lemma}{Lemma}[section]
\newtheorem{prop}[lemma]{Proposition}
\newtheorem{theorem}[lemma]{Theorem}
\newtheorem{cor}[lemma]{Corollary}

\newtheorem{example}[lemma]{Example}
\newtheorem{definition}[lemma]{Definition}
\newtheorem{question}[lemma]
{Question}

\newcommand{\C}{\mathbb{C}}

\def\ket#1{| #1 \rangle}
\def\bra#1{\langle #1 |}
\def\kb#1#2{|#1\rangle\!\langle #2 |}
\def\bk#1#2{\langle #1 \vert #2 \rangle}

\def\be{\begin{eqnarray}}
\def\ee{\end{eqnarray}}
\def\bee{\begin{eqnarray*}}
\def\eee{\end{eqnarray*}}

\def\ot{\otimes}
\renewcommand{\H}{{\mathcal H}}

\begin{document}

\title{Graph Structures for Local Distinguishability of Quantum Product States}

\author{ 
Sooyeong Kim\textsuperscript{1,2}, David W. Kribs\textsuperscript{1}, Michael Nathanson\textsuperscript{3}, Rajesh Pereira\textsuperscript{1}, and Sarah Plosker\textsuperscript{2}
}

\maketitle

\begin{abstract}
We consider the problem of distinguishing sets of quantum product states with local operations and classical communication (LOCC). Recent work has used graph theory to identify sets of product states  distinguishable with one-way LOCC. We extend these efforts to full two-way LOCC, with the first significant analysis of the set of graphs corresponding to bipartite product states that can be distinguished with two-way protocols after finitely many steps. We derive basic closure properties of the set of distinguishable graphs and identify some classes of graphs that guarantee local distinguishability and some graphs that do not. We also include several examples and forward-looking comments.

    \medskip

	

\end{abstract}

\addtocounter{footnote}{1}
\footnotetext{Department of Mathematics \& Statistics, University of Guelph, Guelph, ON, Canada  N1G 2W1}
\addtocounter{footnote}{1}
\footnotetext{Department of Mathematics \& Computer Science, Brandon University, Brandon, MB, Canada R7A 6A9}  
\addtocounter{footnote}{1}
\footnotetext{Department of  Mathematical Sciences, Bentley University, Waltham, MA, USA 02452}


\section{Introduction}

Quantum state distinguishability is a central topic in the theory of quantum communication. A core problem in the subject is to determine the distinguishability (or indistinguishability) of sets of quantum states shared amongst multiple parties. Of particular interest, in both theory and applications, is the paradigm in which each party can implement local (quantum) measurements but only classical communication is available between the parties. This restricted type of communication is denoted LOCC \cite{bennett1999quantum,ghosh2004distinguishability,horodecki2003local,chefles2004condition} and includes key topics such as quantum teleportation, data hiding, and several of their derivations \cite{Teleportation, terhal2001hiding,eggeling2002hiding}. 

Entanglement is the most famous example of quantum weirdness, but it is not necessary to observe quantum phenomena. Quantum (tensor) product states are precisely those pure states that are not entangled, and so, at least in principle, are easier to handle mathematically. Nevertheless, as demonstrated in the seminal quantum communication work \cite{bennett1999quantum}, there exist sets of orthogonal product states, including complete product bases, that cannot be perfectly distinguished with LOCC. This phenomenon is known as `quantum nonlocality without entanglement'.  
This implies that there exist sets of quantum states that can be encoded using only classical communication followed by local quantum operations but which cannot then be decoded using only LOCC. This counter-intuitive phenomenon has motivated interest in determining whether or not a given set of product states is locally distinguishable since the earliest days of quantum information. 

Recently, techniques from graph theory have been brought to the table for this problem \cite{kribs2020vector,kribs2021operator,KMNP2024}. The basic idea is that, in the case of quantum product states distributed over multiple parties, there is an orthogonal representation of the graph assigned to each party by the orthogonality relations on their parts of the states. In a bipartite system, with two parties called Alice and Bob, the product states define orthogonal representations of graphs for Alice's (with graph denoted by $G_A$) and Bob's (with graph denoted by $G_B$) parts of the states. One-way LOCC is a special case of interest that has received considerable attention \cite{Walgate-2000,nathanson2005distinguishing,fan2004distinguishability,Nathanson2013,cosentino2013small,yu2012four,kribs2017operator, kribs2019quantum,kribs2020one}, which in a bipartite system is defined by the requirement that Alice must complete her measurement before Bob can start his. 
In \cite{KMNP2024}, results were obtained for one-way LOCC in a bipartite system, with chordal graphs identified as a key graph structure associated with guaranteeing one-way distinguishability. It was also shown that the graphs on their own do not determine distinguishability for the general case. 

In this paper, we undertake the first comprehensive effort to extend this graph theoretic approach to full two-way LOCC in the bipartite case; that is, scenarios in which Alice and Bob can make local partial measurements and communicate as much as they like. 
Thus, the basic motivating question we consider in this work is the following: 
 
\begin{question}\label{question original} Given a set of bipartite product states $S$ on a finite-dimensional Hilbert space $\mathcal H_A \otimes \mathcal H_B$ with orthogonal graph representations $(G_A, G_B)$, how, and to what extent, can we determine from graph theory whether the elements of $S$ can be distinguished with LOCC? 
\end{question}

We will expand on this question and surrounding details in the next section, which includes relevant preliminaries and an extended background discussion on LOCC state distinguishability. In Section~3, we briefly revisit the one-way case, with an eye toward extending it, and then we define the relevant set of graphs for full LOCC, what we call the {\it distinguishable graphs}, and determine some of its easily identifiable subsets. Section~4 includes our main results and observations on the structure of the set of distinguishable graphs, with an emphasis on identifying its closure properties and related important classes of graphs.  
We also present several examples throughout the work, and in Section~5 we conclude with some forward looking remarks.

\section{Preliminaries and Background}

In this section, we introduce the necessary background material and notation, and then we refine the above motivating question as part of a short introductory discussion on LOCC distinguishability of product states and graph theory. 

\subsection{Quantum Measurements and LOCC}

The quantum systems we consider will be represented by a finite-dimensional Hilbert space $\mathcal H$, which is made up of two subsystems representing the quantum state space of two physically separated parties, call them Alice and Bob. Mathematically, we have the joint system $\mathcal H=\mathcal H_A\otimes \mathcal H_B$ where $\mathcal H_A$ and $\mathcal H_B$ are finite-dimensional complex vector spaces.  A positive operator valued measure (POVM) is a set of positive (semidefinite) operators $E_k$ that sum to the identity. Local operations and classical communication (LOCC) is a natural subclass of quantum operations whereby the parties involved in the communication are restricted to act locally on their respective subsystems (i.e., Alice can only act on $\mathcal H_A$ and similarly for Bob) by performing measurements and
quantum operations, and the 
parties are free to communicate any classical data \cite{chitambar2014everything}. 

Quantum states are unit length vectors, written in Dirac ket form $\ket{\psi}$ (with dual functionals written in bra form as $\bra{\psi}$). Product states in $\mathcal H$ can be written as tensor products of states on the individual systems:  $\ket{\psi}=\ket{\psi_A}\otimes \ket{\psi_B}$. The basic setup in the LOCC state distinguishability framework is that the system $\mathcal H$ has been prepared in a pure state $\ket{\psi_k}$ from a known set of states $\mathcal S=\{\psi_i\}\subseteq \mathcal H$. Alice and Bob would like to determine the value of $k$ using only LOCC. 

One-way LOCC, for example, refers to schemes for which the classical information can only flow in one direction, and full multi-directional LOCC entails no restrictions on the classical communication flows. More precisely, a one-way LOCC measurement with Alice going first (i.e., only Alice can communicate classical information of the two parties) is of the form $\mathbb M = \{ A_k \otimes B_{k,j} \}_{k,j}$, with the $A_k$'s and $B_{k,j}$'s being  positive operators satisfying $\sum_k A_k = I_A$ and $\sum_{j} B_{k,j} = I_B$ for each $k$ \cite{KMNP2024}. Without loss of generality, each of the $A_k$ and $B_{k,j}$ can be assumed to be rank-one \cite{Nathanson2013}. The operators define the one-way protocol as follows: Alice performs her measurement, communicates the result $k$ to Bob, who then performs his measurement associated with the classical index $k$. Within that measurement, the outcome $j$ allows him to identify the state.  

\subsection{Graph Theory Basics}

A  graph $G=(V, E)$ consists of a set $V$ of vertices and a set $E$ of edges (pairs of vertices). If $(u,v)\in E$,   we say that vertices $u$ and $v$ are adjacent in $G$, and we write $u\sim v\in G$. Here, we  consider only simple graphs: unweighted (edges are all assigned the number 1), undirected (edges are unordered pairs) graphs with no loops (edges $(u,v)$ satisfy $u\neq v$). The complement of a graph $G$ is $\overline{G}=(V, \overline{E})$ where $(u,v)\in \overline{E}$ if and only if $(u,v)\notin E$. A graph $G' = (V',E')$ is a subgraph of $G$, denoted $G'\leq G$, provided $V'\subseteq V$, $E'\subseteq E$, and $(u,v)\in E'$ implies $u,v\in V'$. An induced subgraph $G'$  of a graph $G$  is a subgraph for which $u,v, \in V'$ and $(u,v)\in E$ implies that $(u,v) \in E'$ (all edges in $E$ connecting vertices in $V'$ are also edges in $E'$) and is written $G' = G[V']$. A spanning subgraph $G'$ of $G$ satisfies $V'=V$ (but not necessarily $E'=E$). Vertices  $u$ and $v$ are neighbours in $G$ if $(u,v) \in E$; the set of neighbours of $u$ (excluding $u$) is denoted $N_G(u)$; if we wish to include $u$ in its set of neighbours, this is called the closed neighbourhood of $u$, denoted $N_G[u]$. 

The {\it complete graph} on $n$ vertices is the graph $K_n=(\{v_1, \dots v_n\}, E)$ where $(u,v)\in E$ provided $u\neq v$ (i.e., there is an edge between every pair of vertices). The \emph{path} on $n$ vertices is a graph $P_n=(\{v_1, \dots, v_n\}, E)$ for which the $v_i$ are distinct and  $E=\{(v_i, v_{i+1})\,:\, 1\leq i\leq n-1\}$.  Two vertices $u$ and $v$ in a graph $G$ are connected if $G$ contains  a path from $u$ to $v$. A graph is {\it connected} if there is a path between every pair of vertices. If a graph is not connected then it is disconnected. A totally disconnected graph is a graph whose edge set is empty (note that the totally disconnected graph on $n$ vertices, often written as $O_n$,  can also be written as $\overline{K_n}$).  The \emph{cycle} on $n$ vertices is a graph $C_n=(\{v_1, \dots, v_n\}, E)$ for which $E=\{(v_i, v_{i+1})\,:\, 1\leq i\leq n-1\}\cup \{(v_n, v_1)\}$. A chord of a cycle is an edge $(u,v)$ that does not belong to the edge set of the cycle, but $u$ and $v$ are in the vertex set of the cycle.  A \emph{chordal graph} $G$ is a graph in which every cycle $C_n\leq G$ with $n\geq 4$ has a chord. 

A {\it clique} $C$ of a graph $G$ is a subset of vertices of $G$ (i.e., $C\subseteq V$) such that $u, v\in C$, $u\neq v$, implies $(u,v)\in E$. Equivalently, the subgraph of $G$ induced by $C$ is a complete graph.  A {\it coclique}, also called an {\it independent set}, is a subset $S$ of the vertices of a graph $G$ such that $u, v\in S$, $u\neq v$, implies $(u,v)\notin E$. A {\it split graph} is a graph for which the vertices can be partitioned into a clique and a coclique. It can easily be seen that any split graph is chordal.  The converse is known to be false; the path on five vertices is a chordal graph which is not a split graph.
A {\it clique separating set} of $G$ is a subset of vertices such that the induced subgraph defined by the vertices forms a clique and the removal from $G$ of the vertices and the edges that connect to them is a disconnected graph. 
A set of graphs $\{ G_i = (V_i, E_i) \}$ {\it covers} $G$ if $V = \cup_i V_i$ and $E = \cup_i E_i$. A collection of graphs $\{ G_i \}$ is a {\it clique cover} for $G$ if $\{G_i\}$ covers $G$ and if each of the $G_i$ is a complete graph (i.e.,
a clique). The {\it clique cover number} $\mathrm{cc}(G)$ is the smallest possible number of subgraphs contained in a clique cover of $G$. A clique cover can be thought of as a collection of (not necessarily disjoint) induced subgraphs of $G$, each of which is a complete graph with the condition that every edge is contained in at least one of the cliques.

\subsection{Orthogonal Representations and LOCC Protocols}

We now explicitly connect quantum measurements and graph theory: 
\begin{definition}
    For a given graph $G= (V,E)$, a function $\phi: V \rightarrow \mathbb{C}^d\backslash \{0\} $ is an {\it orthogonal representation} of $G$ if for all distinct vertices $u,v\in V$, we have
$$
u \not\sim v \iff \langle \phi(u), \phi(v) \rangle = 0.
$$
\end{definition}
Any set of bipartite product states $\{ \ket{\psi^A_k}\otimes \ket{\psi^B_k} \}_{k=1}^r$ on $\mathcal H_A \otimes \mathcal H_B$ can be associated with a pair of graphs. The orthogonality graph from Alice's perspective, called \emph{Alice's graph}, is the unique graph $G_A$ with vertex set $V = \{ 1,2, \ldots, r\}$ such that the map $k \mapsto  \ket{\psi^A_k}$ is an orthogonal representation of $G_A$. Bob's graph $G_B$ can be defined similarly. Observe that the product states are mutually orthogonal if and only if $G_A \le \overline{G_B}$, or equivalently vice-versa. 

Continuing our discussion from above, note that given any pair of graphs $(G_A, G_B)$ with $G_A \le \overline{G_B}$, there always exist corresponding sets of product states that can be distinguished with local operations and \emph{no} mid-measurement communication: we can see this by constructing a vertex clique cover of $G_A$ and identifying each clique with an element of the standard basis. We then let Alice's states be the characteristic vectors of each vertex. We can build a similar representation for Bob's graph. If Alice and Bob each measure in the standard basis, then their pair of outcomes $(a,b)$ uniquely determines the original state. (This process can be easily extended to multipartite systems.)

As a result of this observation, we can refine Question~\ref{question original} in the context of full LOCC. 
\begin{question}
Are there pairs of graphs $G_A \le \overline{G_B}$ such that \emph{every} set of orthogonal bipartite product states $S$ with  representation $(G_A, G_B)$  is distinguishable with (full) LOCC?
\end{question}
\begin{question}\label{question indistinguishable}
Are there pairs of graphs  $G_A \le \overline{G_B}$ and dimensions $d_A$ and $d_B$ such that \emph{no} set of orthogonal bipartite product states $S \subset \C^{d_A} \ot \C^{d_B}$ with representation $(G_A, G_B)$  is distinguishable with LOCC?
\end{question}

For the benefit of readers with limited exposure to the topic, we discuss LOCC protocols in more detail, at the level of explicit measurement steps and classical communication. Our aim is to give an intuitive understanding of how one-way and two-way LOCC discrimination proceeds in practice, via presentation of an example, before introducing formal criteria. 
Recall that a set of product states 
$
\{\,\ket{\psi_k^A} \otimes \ket{\psi_k^B}\,\}_{k=1}^r
$
is distinguishable by LOCC if Alice and Bob can identify the given state with certainty using only local measurements and classical communication. Further, Alice and Bob are assumed to know the entire set of states in advance and to agree beforehand on a fixed LOCC protocol, including the order of measurements and all conditional measurement choices based on classical outcomes. The set is said to be distinguishable if, for every $k$, the protocol identifies the state $\ket{\psi_k^A} \otimes \ket{\psi_k^B}$ with probability $1$.

\begin{example}\label{1way2wayeg}
{\rm 
Consider the following five orthogonal product states in $\mathbb{C}^2 \otimes \mathbb{C}^3$, where $\ket{\pm} = (\ket{0} \pm \ket{1})/\sqrt{2}$:
\begin{equation*}
\begin{aligned}
\psi_1 = \ket{1} &\otimes \ket{1}, \quad
\psi_2 = \ket{0} \otimes \ket{+}, \quad
\psi_3 = \ket{0} \otimes \ket{-}, \\
&\psi_4 = \ket{+} \otimes \ket{2}, \quad
\psi_5 = \ket{-} \otimes \ket{2} .
\end{aligned}
\end{equation*}
Alice and Bob are given one of these states and must identify it using only LOCC.

As we have discussed above, we may associate graphs to Alice and Bob in which vertices correspond to indices of product states (or local states) and edges represent local non-orthogonality relations among them. See Figure~1 for a graph representation of the example. As noted in the next section, the resulting pair of graphs captures which pairs of states can or cannot be distinguished locally. Before turning to that discussion, let us explicitly describe the different types of distinguishability that apply in this example.  

\begin{figure}[h!]\label{fig:alice bob graphs}
	\begin{center}
		\begin{tikzpicture}[scale=1.3, baseline=(current bounding box.north)]
		\draw[thick] (-1,0.4)--(1,0.4);
		
		\draw[thick] (-1,0.4)--(-0.6,-1);
		\draw[thick] (-1,0.4)--(0.6,-1);
		
		\draw[thick] (1,0.4)--(-0.6,-1);
		\draw[thick] (1,0.4)--(0.6,-1);
		
		\draw[thick] (0,1)--(-0.6,-1);
		\draw[thick] (0,1)--(0.6,-1);
		
		\draw[fill] (0,1) node[anchor=south] {$1$} circle [radius=2pt];
		\draw[fill] (-1,0.4) node[anchor=east] {$2$} circle [radius=2pt];
		\draw[fill] (1,0.4) node[anchor=west] {$3$} circle [radius=2pt];
		\draw[fill] (-0.6,-1) node[anchor=north] {$4$} circle [radius=2pt];
		\draw[fill] (0.6,-1) node[anchor=north] {$5$} circle [radius=2pt];
		
		\node at (0,1.8) {Alice graph $G_A$};
		
		\begin{scope}[xshift=5.5cm]

		\draw[thick] (0,1)--(-1,0.4);
		\draw[thick] (0,1)--(1,0.4);
		\draw[thick] (-0.6,-1)--(0.6,-1);

		\draw[fill] (0,1) node[anchor=south] {$1$} circle [radius=2pt];
		\draw[fill] (-1,0.4) node[anchor=east] {$2$} circle [radius=2pt];
		\draw[fill] (1,0.4) node[anchor=west] {$3$} circle [radius=2pt];
		\draw[fill] (-0.6,-1) node[anchor=north] {$4$} circle [radius=2pt];
		\draw[fill] (0.6,-1) node[anchor=north] {$5$} circle [radius=2pt];
		
		\node at (0,1.8) {Bob graph $G_B$};
		
		\end{scope}
		
		\end{tikzpicture}
	\end{center}
	\caption{Graph representation of local orthogonality relations in Example~\ref{1way2wayeg}.}
\end{figure}
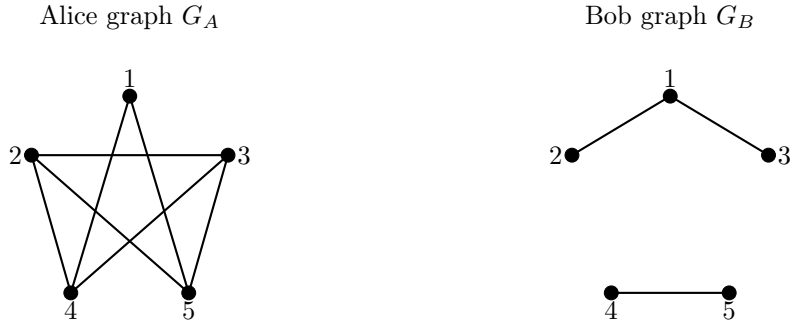

This set can be distinguished by a Bob-first one-way LOCC protocol. Bob first performs a local projective measurement in the basis $\{\ket{+},\ket{-},\ket{2}\}$. If the outcome is $\ket{2}$, the remaining possible states are $\psi_4$ and $\psi_5$, which Alice can distinguish by measuring in the basis  $\{\ket{+},\ket{-}\}$. If the outcome is $\ket{+}$, the remaining possible states are $\psi_1$ and $\psi_2$; if the outcome is $\ket{-}$, the remaining possible states are $\psi_1$ and $\psi_3$. In either case, Alice can distinguish the remaining states by measuring in the computational basis $\{\ket{0},\ket{1}\}$.

This set can also be distinguished via a Bob-first two-way LOCC protocol. Bob first performs a local projective measurement with outcomes corresponding to the operators $\ket{2}\bra{2}$ and $\ket{0}\bra{0} + \ket{1}\bra{1}$
(from Bob’s perspective, this measurement leaves all states unchanged; other measurements could change the post-measurement states). Bob then communicates the outcome to Alice. If Bob’s outcome corresponds to $\ket{2}\bra{2}$, the state must be either $\psi_4$ or $\psi_5$, which Alice distinguishes by measuring in the basis $\{\ket{+},\ket{-}\}$. If Bob’s outcome instead corresponds to the subspace spanned by $\ket{0}$ and $\ket{1}$, the remaining possible states are $\psi_1$, $\psi_2$, and $\psi_3$. Alice then measures in the basis $\{\ket{0},\ket{1}\}$. The outcome $\ket{1}$ identifies $\psi_1$, while the outcome $\ket{0}$ leaves $\psi_2$ and $\psi_3$ indistinguishable to Alice. In this case, Alice communicates her outcome to Bob, who performs a final measurement in the basis  $\{\ket{+},\ket{-}\}$ to distinguish $\psi_2$ from $\psi_3$.

In an Alice-first one-way LOCC protocol, each outcome of Alice’s measurement must leave Bob with states that are perfectly distinguishable. Since $\psi_4$ and $\psi_5$ have identical states on Bob’s subsystem, no outcome can contain both, so each outcome must exclude either $\ket{+}$ or $\ket{-}$ on Alice’s side. Moreover, since Bob cannot distinguish $\ket{1}$ from $\ket{+}$ or $\ket{-}$, Alice cannot allow $\ket{0}$ and $\ket{1}$ to appear together in any outcome. However, any outcome that contains $\ket{1}$ while excluding $\ket{0}$ necessarily contains both $\ket{+}$ and $\ket{-}$. This forces $\psi_4$ and $\psi_5$ to remain in the same outcome, a contradiction. Hence, Alice-first one-way LOCC is impossible for this set. Nevertheless, as noted above, the set is distinguishable by Alice-first two-way LOCC since Bob can initiate a one-way protocol while Alice takes no initial action.

}
\end{example}

\section{From One-Way to Full LOCC Distinguishability}

In this section, we first present some recent results and a series of examples on one-way LOCC distinguishability, with an eye toward extending the graph theory approach to full LOCC. We then define and obtain initial results on the set of distinguishable graphs.

\subsection{Structural Conditions for One-Way LOCC}	

As described above, orthogonal (bipartite) quantum product states can be described by graph representations that capture non-orthogonality relations on Alice’s and Bob’s subsystems. This raised a structural question: given only the associated graphs, can we determine whether the set of states is perfectly distinguishable without specifying the concrete states? This has been studied in the context of one-way LOCC distinguishability \cite{kribs2020vector,kribs2021operator,KMNP2024}, where the answer was found at the highest level to be: sometimes.

An early result in LOCC distinguishability \cite{bennett1999unextendible,Walgate-2002,divincenzo2003unextendible,halder2020distinguishability} shows that any set of orthogonal product states in $\mathbb{C}^{2} \otimes  \mathbb{C}^{d}$ for arbitrary $d\geq 2$ can be distinguished  via full (two-way) LOCC. As proved in \cite{Walgate-2002}, the states can always be written in a form in terms of orthogonality on  Alice's side versus the corresponding orthogonalities on Bob's side that allow for distinguishability. On the other hand, for one-way LOCC this `single-qubit sender' scenario can result in states that are indistinguishable as well. In fact, this case is one in which we have a complete graph theoretic characterization. 

\begin{theorem}\cite{kribs2021operator,KMNP2024}\label{singlequbit}
	A set of orthonormal product states in $\mathbb{C}^{2} \otimes  \mathbb{C}^{d}$, for~$d \geq 2$, is distinguishable via one-way LOCC with Alice going first if and only if there is some graph between the two graphs $G_A$ and $\overline{G_B}$ with a clique cover number of at most two; that is, there is a graph $G$ such that
	\begin{equation*}
	G_A \leq G \leq \overline{G_B} \quad \mathit{and} \quad \mathrm{cc}\,(G) \leq 2.
	\end{equation*}
\end{theorem}

Let us return to the example above and observe how the distinguishability conclusions there follow readily from this result. 

\begin{example}
{\rm 
As one can easily see from Figure~1, in the case of Example~\ref{1way2wayeg} 
we have $G_B = \overline{G_A}:=G$ and $\mathrm{cc}\,(G) = 2$. Thus it follows immediately from the result that the states are one-way distinguishable with Bob going first (as observed in the example discussion). On the other hand, $\mathrm{cc}\,(\overline{G}) = \mathrm{cc}\,(G_A) = 4$, and so the states are not one-way distinguishable with Alice going first, which was also observed above (even though they are two-way distinguishable with Alice going first). 
}
\end{example} 

More generally, the following graph-theoretic description has been obtained, linking chordal graph structure with one-way distinguishability of {\it all} orthogonal representations for which the graph lies between the associated Alice and Bob (complement) graphs. 

\begin{theorem}\label{chordalthm}\cite{KMNP2024}
	Let $G$ be a graph with $n$ vertices. Then $G$ is chordal if and only if every collection of $n$ product states having Alice graph $G_A$ and Bob graph $G_B$ with
	\[
	G_A \le G \le \overline{G_B}
	\]
	is distinguishable with one-way LOCC with Alice going first.
\end{theorem}

This result has a number of consequences, such as the previous theorem which can be derived from it.   
However, we note that the result has some embedded subtleties, associated jointly with the graphs and the state representations, that warrant further investigation even just in the one-way case. For instance, 
the following example shows that one-way LOCC distinguishability can still hold even when the associated graph is not chordal.

\begin{example}
{\rm 
	Consider the (unnormalized) set of product states $\{\ket{\psi_i}\}_{i=1}^4 \subset \mathbb{C}^4 \otimes \mathbb{C}^2$ defined by
	\[
	\begin{aligned}
	\ket{\psi_1} &= (\ket{0} + \ket{1}) \otimes \ket{0}, &
	\ket{\psi_2} &= (\ket{1} + \ket{2}) \otimes \ket{1}, \\
	\ket{\psi_3} &= (\ket{0} + \ket{3}) \otimes \ket{1}, &
	\ket{\psi_4} &= (\ket{2} + \ket{3}) \otimes \ket{0}.
	\end{aligned}
	\]
	Here $G_A$ is the four-cycle $C_4$, and $G_B$ is its complement with edge set $\{(1,4),(2,3)\}$. So $G_A = \overline{G_B}$, and this graph is not chordal. Nevertheless, the states are distinguishable by one-way LOCC with Alice going first; indeed, Alice can measure in the computational basis $\{\ket{0},\ket{1},\ket{2},\ket{3}\}$, which reduces the problem to a pair of states, and Bob then completes the discrimination by measuring in the basis $\{\ket{0},\ket{1}\}$.
    }
\end{example}

We also have access to a number of results that tell us states are one-way indistinguishable when certain graph constraints are satisfied. As an example, we note the following result. 

Let $G$ be a graph with $n = |G|$. We write $S_G$ for the set of $n \times n$ matrices whose zero–nonzero pattern is constrained by the adjacency structure of $G$. The positive semidefinite minimum rank of $G$ was defined in \cite{KMNP2024} by
\[
\eta_{+}(G) := \min \{ \mathrm{rk}(M) : M \in S_G, M \textnormal{ is positive semidefinite with invertible diagonal}\}.
\]
Recall that a {\it simplicial vertex} in a graph is a vertex whose neighbours form a clique.

\begin{theorem}\cite{KMNP2024}\label{thm:min rank impossible one way}
	Let $S$ be a collection of orthonormal product states in 
	$\mathcal{H}_A \otimes \mathcal{H}_B$, with associated graphs 
	$G_A$ and $G_B$, and suppose that $G = \overline{G_B}$ has no simplicial vertices. If $\dim \mathcal{H}_A = \eta_{+}(G)$, then the product states of $S$ 
	cannot be distinguished via one-way LOCC with Alice going first.
\end{theorem}

The following is an illustrative example of this phenomenon. 

\begin{example}\label{ex:not one way but three way}
{\rm 
	Consider the following set of (unnormalized) states on $\mathbb C^3 \otimes \mathbb C^3$:

    \begin{equation*}
\begin{aligned}
	\psi_1 = \ket{0} \otimes & \ket{0+1},\quad \psi_2 = \ket{0} \otimes \ket{0-1}, \quad \psi_3= \ket{2} \otimes \ket{1+2},\quad \psi_4 = \ket{2} \otimes \ket{1-2}, \\
	& \psi_5 = \ket{1-2} \otimes \ket{0}, \quad \psi_6 = \ket{0+1} \otimes \ket{2},\quad
	\psi_7 = \ket{0-1} \otimes \ket{2} .
\end{aligned}
\end{equation*}

	\begin{figure}[h!]\label{fig:alice-bob-localstate-graphs}
		\begin{center}
			\begin{tikzpicture}[scale=1.15, baseline=(current bounding box.north)]
			
			
			\draw[fill] (-2,1)   node[anchor=south] {$1$}     circle (2pt);
			\draw[fill] (-2,0)   node[anchor=north] {$2$}     circle (2pt);
			
			\draw[fill] (-0.7,1) node[anchor=south] {$6$}   circle (2pt);
			\draw[fill] (-0.7,0) node[anchor=north] {$7$}   circle (2pt);
			
			\draw[fill] (0.7,0)  node[anchor=north] {$5$}   circle (2pt);
			
			\draw[fill] (2,0)    node[anchor=north] {$3$}     circle (2pt);
			\draw[fill] (2,1)    node[anchor=south] {$4$}   circle (2pt);
			
			\draw[thick] (-2,1)--(-2,0);
			\draw[thick] (-2,1)--(-0.7,1);
			\draw[thick] (-2,1)--(-0.7,0);
			
			\draw[thick] (-2,0)--(-0.7,0);
			\draw[thick] (-2,0)--(-0.7,1);
			
			\draw[thick] (-0.7,1)--(0.7,0);
			
			\draw[thick] (-0.7,0)--(0.7,0);
			
			\draw[thick] (0.7,0)--(2,1);
			\draw[thick] (0.7,0)--(2,0);
			\draw[thick] (2,1)--(2,0);
			
			\node at (0,1.8) {Alice graph $G_A$};
			
			\begin{scope}[xshift=7.2cm]
			
			\draw[fill] (-2,1)   node[anchor=south] {$6$}     circle (2pt);
			\draw[fill] (-2,0)   node[anchor=north] {$7$}     circle (2pt);
			
			\draw[fill] (-0.7,1) node[anchor=south] {$3$}   circle (2pt);
			\draw[fill] (-0.7,0) node[anchor=north] {$4$}   circle (2pt);
			
			\draw[fill] (0.7,1)  node[anchor=south] {$1$}   circle (2pt);
			\draw[fill] (0.7,0)  node[anchor=north] {$2$}   circle (2pt);
			
			\draw[fill] (2,0.5)  node[anchor=west] {$5$}      circle (2pt);

			\draw[thick] (-2,1)--(-2,0);
			\draw[thick] (-2,1)--(-0.7,1);
			\draw[thick] (-2,0)--(-0.7,0);
			\draw[thick] (-2,1)--(-0.7,0);
			\draw[thick] (-2,0)--(-0.7,1);
			
			\draw[thick] (-0.7,1)--(0.7,1);
			\draw[thick] (-0.7,0)--(0.7,0);
			
			\draw[thick] (2,0.5)--(0.7,0);
			\draw[thick] (2,0.5)--(0.7,1);

			\node at (0,1.8) {Bob graph $G_B$};
			
			\end{scope}
			\end{tikzpicture}
		\end{center}
			\caption{Graph representation of local orthogonality relations in Example~\ref{fig:alice-bob-localstate-graphs}.}
	\end{figure}
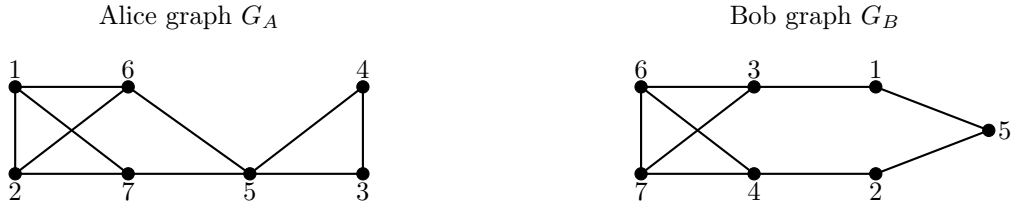
	Observe that $G = \overline{G_B}$ has no simplicial vertices. Here we have  $\dim\mathcal H_A= \dim\mathcal H_B=3$, and one can show (by considering basic graph parameter inequalities) that $\eta_+(G)=3$. Hence, Theorem~\ref{thm:min rank impossible one way} implies that the states are not distinguishable via one-way LOCC with Alice going first. Applying the same reasoning to $\overline{G_A}$ shows that one-way LOCC distinguishability with Bob going first is also impossible. We will revisit this example in the next section, showing how the states are in fact distinguishable with full LOCC. 
    }
\end{example}

We point the reader to \cite{KMNP2024} for various other examples of the fact that there are many graphs with representations that can always be distinguished with two-way LOCC but not always with one-way.

\subsection{Full LOCC Distinguishability and Graph Theory}

We next extend the graph-theoretic framework adapted to full LOCC. To this end, let us conceptually illustrate how an LOCC protocol can be described via a tree representation that is often used to describe LOCC protocols, making use of the example above.

\begin{example} 
{\rm 
Returning to Example~\ref{ex:not one way but three way}, recall the states are not distinguishable via one-way LOCC, regardless of which party goes first. They can, however, be perfectly distinguished using a multi-round LOCC protocol. The full LOCC protocol is depicted in Figure~\ref{fig:locc-tree-example} as a measurement tree, following the construction in \cite{bennett1999quantum}.

\begin{figure}[h!]
	\centering
	\begin{tikzpicture}[scale=1.0, line cap=round, line join=round]
	
	\tikzset{
		dot/.style={circle, fill=black, inner sep=1.6pt},
		lab/.style={font=\large},
		leaf/.style={font=\large\bfseries},
		lvl/.style={dashed, gray},
	}
	
	\def\yone{0}
	\def\ytwo{-2}
	\def\ythree{-4}
	\def\yfour{-6}
	
	\draw[lvl] (-6.5,\yone) -- (6,\yone);
	\draw[lvl] (-6.5,\ytwo) -- (6,\ytwo);
	\draw[lvl] (-6.5,\ythree) -- (6,\ythree);
	\draw[lvl] (-6.5,\yfour) -- (6,\yfour);
	
	\node[lab] at (-7, -1.0) {\itshape 1};
	\node[lab] at (-7, -3.1) {\itshape 2};
	\node[lab] at (-7, -5.3) {\itshape 3};
	
	\node[dot] (r)   at (0,\yone) {};
	
	\node[dot] (a01) at (-3.2,\ytwo) {};
	\node[dot] (a2)  at ( 3.2,\ytwo) {};
	
	\node[lab] at (-2,-0.9) {\normalsize A0/1};
	\node[lab] at ( 2,-0.9) {\normalsize A2};
	
	\draw (r) -- (a01);
	\draw (r) -- (a2);
	
	\node[dot] (b02) at (-5,\ythree) {};
	\node[dot] (b01) at (-2.6,\ythree) {};
	\node[dot] (b2)  at (-0.2,\ythree) {};
	
	\node[dot] (L5a) at ( 1.2,\ythree) {};
	\node[dot] (L3)  at ( 3.4,\ythree) {};
	\node[dot] (L4)  at ( 4.8,\ythree) {};
	
	\node[lab] at (1.4,-3.3) {\normalsize B0};
	\node[lab] at (2.7,-3.3) {\normalsize B(1$+$2)};
	\node[lab] at (5,-3.3) {\normalsize B(1$-$2)};
	
	\node[lab] at (-5.2,-3.3) {\normalsize B(0$+$1)};
	\node[lab] at (-3.43,-3.3) {\normalsize B(0$-$1)};
	\node[lab] at (-0.7,-3.3) {\normalsize B2};
	
	\draw (a2) -- (L5a);
	\draw (a2) -- (L3);
	\draw (a2) -- (L4);
	
	\draw (a01) -- (b01);
	\draw (a01) -- (b02);
	\draw (a01) -- (b2);

	\node[dot] (a0p) at (-5.7,\yfour) {};
	\node[dot] (a1p) at (-4.3,\yfour) {};
	
	\node[dot] (a0) at (-3.2,\yfour) {};
	\node[dot] (a1) at (-1.9,\yfour) {};

	\node[dot] (L6)  at (0,\yfour) {};
	\node[dot] (L7)  at (1.3,\yfour) {};
	
	\node[lab] at (-5.8,-5.3) {\normalsize A0};
	\node[lab] at (-4.2,-5.3) {\normalsize A1};
	\node[lab] at (-3.3,-5.3) {\normalsize A0};
	\node[lab] at (-1.8,-5.3) {\normalsize A1};
	\node[lab] at (-.7,-5.3) {\normalsize A(0$+$1)};
	\node[lab] at (1.5,-5.3) {\normalsize A(0$-$1)};
	
	\draw (b02) -- (a0p);
	\draw (b02) -- (a1p);
	
	\draw (b01) -- (a0);
	\draw (b01) -- (a1);
		
	\draw (b2) -- (L6);
	\draw (b2) -- (L7);

%
%
%
%

	\node[leaf] at ( 1.2,\ythree-0.55) {5};
	\node[leaf] at ( 3.4,\ythree-0.55) {3};
	\node[leaf] at ( 4.8,\ythree-0.55) {4};
	
	\node[leaf] at (-5.7,\yfour-0.55) {1};
	\node[leaf] at (-4.3,\yfour-0.55) {5};
	\node[leaf] at (-3.2,\yfour-0.55) {2};
	\node[leaf] at (-1.9,\yfour-0.55) {5};
	\node[leaf] at (0,\yfour-0.55) {6};
	\node[leaf] at (1.3,\yfour-0.55) {7};
	

	\end{tikzpicture}
	\caption{Measurement tree for a two–way LOCC protocol in Example~\ref{ex:not one way but three way} with Alice going first.}
	\label{fig:locc-tree-example}
\end{figure}
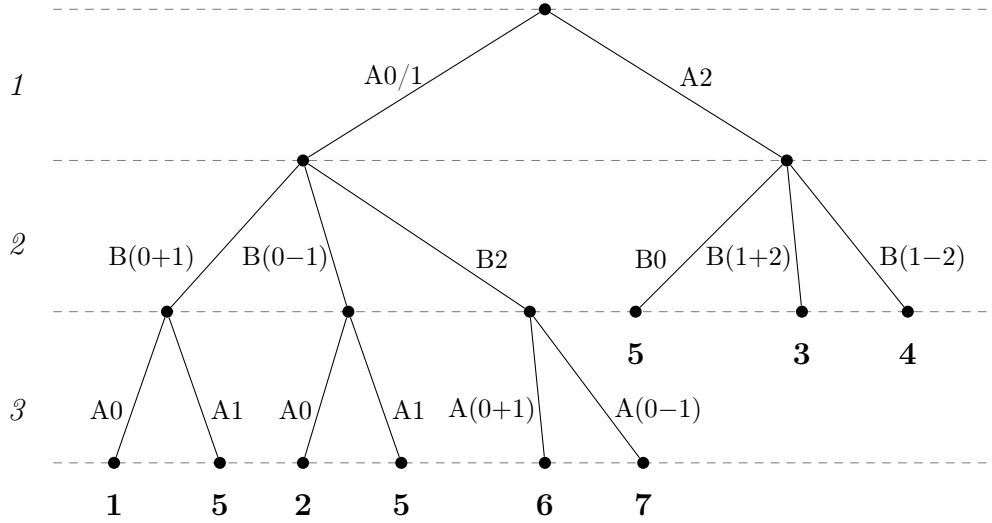

Regarding the labels in Figure~\ref{fig:locc-tree-example}, the notation A$(\cdot)$ and B$(\cdot)$ indicate measurements by Alice and Bob, respectively. A label such as A$0/1$ means that the $\ket{0}$ and $\ket{1}$ outcomes are not distinguished. Bold-faced numbers at the leaves indicate states inferred from the measurement chain. Branches at each node correspond to POVM outcomes, so outcomes with zero support (e.g., $\ket{2}$ in Alice’s Level~3 measurement in the basis $\{\ket{0},\ket{1},\ket{2}\}$) is omitted. Note that after Alice’s first projective measurement, $\psi_5 = \ket{1-2}\otimes \ket{0}$ is projected onto $\ket{1} \otimes \ket{0}$ in the $\{0,1\}$ branch.

Explicitly, the protocol works as follows: Alice first performs a projective measurement separating the subspaces $\mathrm{span}\{\ket{0},\ket{1}\}$ and $\mathrm{span}\{\ket{2}\}$. Conditioned on the outcome, Bob applies a measurement adapted to the remaining candidate states. In the $\{0,1\}$ branch, further rounds of local measurements are required, whereas the $\{2\}$ branch is done immediately. The tree encodes the complete sequence of conditional measurements leading to discrimination. 
}
\end{example}

In general, given a set of product states $\{ \ket{\psi^A_k}\otimes \ket{\psi^B_k} \}_{k=1}^r$ on $\mathcal H_A \otimes \mathcal H_B$, the set is perfectly distinguishable via a multi-round LOCC protocol if the associated measurement tree satisfies: 
	\begin{itemize}
		\item[(i)] Each leaf corresponds to a singleton set of states.
		\item[(ii)] For every state, there exists a path from the root to a leaf labelled by that state.
	\end{itemize}

The tree description is conceptually useful for visualizing LOCC protocols. Turning now to our investigation, we recall the content of Theorem~\ref{chordalthm}, namely that every set of bipartite product states represented by the graphs $G_A, G_B$ can be distinguished with one-way LOCC if and only if there exists a graph $G$ such that $G_A \le G \le \overline{G_B}$ and $G$ is chordal. This motivates our starting point for full LOCC as follows.

\begin{definition}\label{Bipartite G Definition}
Define ${\mathcal G}$ to be the set of simple finite graphs $G$ such that every set of bipartite product states represented by graphs $G_A, G_B$ with $G_A\le G \le \overline{G_B}$ can be distinguished with LOCC in a finite number of steps, with the convention that Alice always has the first move. We call $\mathcal G$ the set of {\bf distinguishable graphs} for $A$ and $B$. 

We further define $\mathcal{G}_0$ and ${\mathcal G}_{1}$ to be the subsets of ${\mathcal G}$ that can be distinguished with product measurements and one-way LOCC, respectively.  
\end{definition}
 Note that if we have a set of bipartite product states $\mathcal{S}$ represented by graphs $G_A$ and $G_B$, then the number of vertices in $G_A$ and in $G_B$ are equal to $|\mathcal{S}|$. The fact that $|G_A| = |G_B|$ combined with the inequality $G_A \le G \le \overline{G_B}$ implies that $|G| = |G_A|$ and that every vertex in $G$ is also a vertex in $G_A$. In other words, $G_A$ is a spanning subgraph of $G$ and, similarly, $G_B$ is a spanning subgraph of $\overline{G}$.

We begin by showing that the sets $\mathcal{G}_0$ and ${\mathcal G}_{1}$ can be identified with well-known classes of graphs.

\begin{theorem}\label{splitchordalthm}
    The set $\mathcal{G}_0$ is the set of split graphs, and the set $\mathcal{G}_1$ is the set of chordal graphs. 
\end{theorem}
\begin{proof}
The set ${\mathcal G}_1$ has been proved to be the set of chordal graphs as noted above.

Regarding the set $\mathcal G_0$, note that 0-direction LOCC implies that states can be distinguished with product measurements; i.e., with no communication between Alice and Bob during the measurement. In particular, this implies that one-way LOCC is possible in either direction, which means that both $G$ and $\overline{G}$ must be chordal. One of the equivalent descriptions of a split graph is that it is chordal and its complement is also chordal. Thus, ${\mathcal G}_0$ is contained in the set of split graphs. 

On the other hand, a split graph $G$ has the property that its vertices can be partitioned into sets $V_1$ and $V_2$, where the vertices of $V_1$ form a clique and the vertices of $V_2$ form an independent set. (Any configuration of edges connecting the two sets is allowed.) By definition, every vertex in $V_2$ is a simplicial vertex in $G$ and every vertex in $V_1$ is a simplicial vertex in $\overline{G}$. There will always be exactly one vertex in their intersection, identifying their state: If $x \in V_1$ and $y \in V_2$,  $N_G[y] \cap N_{\overline{G}}[x]$ consists of either $x$ or $y$, depending on which graph contains the edge $(x,y)$. 
This shows that $\mathcal{G}_0$ contains the set of split graphs and hence the two sets coincide.  
\end{proof}

Before continuing deeper into the graph structures that belong to the set of distinguishable graphs, we show cycles with five or more vertices do not belong to the set. 
Recall that an unextendible product basis is a set of pairwise orthogonal product vectors that does not span all of $\bigotimes_{i = 1}^m \mathcal H_i$ and whose orthogonal complement contains no product vector. 
It is well-known that an unextendible product basis cannot be distinguished even with separable measurements. This means that any graph that admits an unextendible product basis is not in $\mathcal{G}$. 

We note that \cite[Lemma 1]{bennett1999unextendible} gives a condition on the dimension of the local spans that determines whether a basis is extendible as follows.

\begin{lemma} \cite{bennett1999unextendible} Let $S=\{\ket{\psi_j}=\bigotimes_{i=1}^m\ket{\phi_{i,j}}\}_{j=1}^n$ be a set of pairwise orthogonal product states in $\bigotimes_{i = 1}^m \mathcal H_i$.  Then $S$ is extendible if and only if there exists a disjoint partition of $S$ into $S_1,S_2,...,S_m$ such that $\dim(\mathrm{span} (\{\ket{\phi_{i,j}}\}_{j:\ket{\psi_j}\in S_i})) < \dim(\mathcal H_i) $ for all  $1\le i\le m$.
\end{lemma}

In the bipartite case, we need to be able to partition our states into $S = S_1 \cup S_2$ such that the dimension of Alice's states in $S_1$ is strictly less than Alice's dimension and the dimension of Bob's states in $S_2$ is strictly less than Bob's dimension in order for these states to be extendible.  In particular, a simple counting argument shows that a set of $n$ product states in $\mathbb C^{d_1}\ot \mathbb C^{d_2}$, with $n\geq \max\{d_1,d_2\}$, is unextendible if every set of $d_1$ Alice states is linearly independent, every set of $d_2$ Bob states is linearly independent, and $d_1 + d_2  \le n +1$. 

Relatedly, we note that Alon and Lovasz \cite{alon2001unextendible} looked at the connectivity of graphs in terms of guaranteeing a representation in general position (i.e., any subset of the represented states with cardinality equal to the dimension of the Hilbert space is linearly independent), though their definition of a representation is slightly different. 

The unextendible product basis viewpoint and the dimension related observations above allow us to prove the following. 

\begin{prop}\label{unextendible}
For every $n \ge 5$, there exists an unextendible product basis whose graphs are $G_A = \overline{G}_B = C_n$. Thus, for any $n \ge 5$, $C_n \notin \mathcal{G}$.
\end{prop}

\begin{proof}
Every $n$-cycle $C_n$ has a general-position representation in dimension $\dim \mathcal H = n-2$. For example, for $n=5$ we can take on $\mathcal H = \mathbb C^3$ the five states 
\[
S= \{ \ket{0}, \ket{0} + \ket{1}, \ket{1}, \ket{1} + \ket{2}, \ket{0} + \ket{2} \}, 
\]
and observe that any two of these states are orthogonal if and only if they are neither adjacent in the list nor are the first and last elements in the list.
Further, every complement $\overline{C_n}$ has a general position representation in dimension $3$ (as long as $n \ge 5$).  We first note that if $n=5$ or $n\ge 7$, there exists a primitive $n$th root of unity $\eta$ with negative real part. When $n=5$, we can take the primitive fifth root of unity $\eta=e^{\frac{4\pi i}{5} }$. It was shown by Gauss that the sum of the primitive $n$th roots of unity is always either $1$, $0$ or $-1$ \cite[Article 81]{gauss1870disquisitiones}.    The sum of the two primitive roots of $n$th roots of unity $e^{\frac{2\pi i}{n}}+ e^{\frac{-2\pi i}{n}}$ is $2\cos( \frac{2\pi}{n})$ which is always a real number strictly greater than one when $n\geq 7$. Hence it follows from Gauss' result that there must be a primitive root of unity with negative real part when $n\geq 7$.

Indeed, 
when either $n=5$ or $n\ge 7$, we have the standard representation of $\overline{C_n}$ in $\mathbb{R}^3$ given by
\bee
\ket{\psi_k} = (\Re(\eta^k), \Im(\eta^k),h) \quad \mathrm{for}\, 1 \leq k \leq n,
\eee
where $\eta$ is a primitive $n$th root of unity such that $\Re(\eta) < 0$ and $h = \sqrt{-\Re(\eta)}$.  It is easily checked that $\bk{\psi_j}{\psi_k} = 0$ if and only if $j-k = \pm 1$.  Unfortunately, there is no primitive sixth root of unity $\eta$ with $\Re(\eta) < 0$ because both of the primitive sixth roots of unity have real parts equal to $1/2$. Therefore, we need a separate  construction for the $n=6$ case.  Consider the following  set of (unnormalized) states:

\[
S= \{ (1,0,0), (0,1,0), (5,0,3),  (3,4,-5), (3,4,5), (0,5,-4) \}.
\]

Define the representations given above of $C_n$ for Alice and of $\overline{C_n}$ for Bob, and so $d_1 = n-2$ and $d_2=3$. Then the corresponding $n$ product states satisfy the conditions identified in the discussion above, with $d_1 + d_2 = n-2 +3 = n + 1$. Thus, the states form an unextendible product basis with the required graph representation conditions for the Alice and Bob graphs, and the result follows. 
\end{proof}

\section{Structure of the set of Distinguishable Graphs}\label{sec:structure}

In previous work, the set of graphs associated with one-way LOCC distinguishability was identified to be exactly the set of chordal graphs \cite{KMNP2024}. In this section, we investigate closure properties of  the set of distinguishable graphs and relate $\mathcal{G}$  to well-known graph classes. 

We start with some easily derived closure properties, which will be useful in what follows. 
Recall that for two graphs $G_1 = (V_1,E_1)$ and $G_2 = (V_2,E_2)$, the disjoint union $G = G_1 \sqcup G_2$ is the graph with vertices $V = V_1 \cup V_2$ and edges $E = E_1 \cup E_2$. The join $G = G_1 \vee G_2$ is the disjoint union of $G_1$ and $G_2$ together with all edges joining $V_1$ and $V_2$. 

\begin{lemma}\label{Lemma First Closure}
The set of distinguishable graphs $\mathcal G$ has the following closure properties. For graphs $G, G_1, G_2$ in $\mathcal{G}$, we have: 
\begin{itemize}
\item{} $\overline{G} \in \mathcal{G}$, 
\item{} $G_1 \sqcup G_2 \in \mathcal{G}$,
\item{} $G_1 \vee G_2 \in \mathcal{G}$.
\end{itemize}
\end{lemma}

\begin{proof}
For the first property, note that shifting from $G$ to $\overline{G}$ is equivalent to interchanging the roles of Alice and Bob: Start with any set of bipartite states whose graphs $(G_A, G_B)$ satisfy $G_A \le \overline{G} \le \overline{G_B}$. 
Then $G_B  \le G \le \overline{G_A}$. Since $G \in \mathcal{G}$, we know that the states represented by $(G_B,G_A)$ can be distinguished with LOCC. Because we allow communication in both directions, interchanging the roles of $G_A$ and $G_B$ does not change this. Hence $\overline{G} \in \mathcal{G}$. 

Next consider the disjoint union $G_1 \sqcup G_2$. Given a  set of bipartite states whose graphs $(G_A, G_B)$ satisfy $G_A \le (G_1 \sqcup G_2 ) \le \overline{G_B}$, then Alice's states corresponding to $V_1$ and $V_2$ lie in orthogonal subspaces and she can perform a measurement to distinguish them without otherwise disturbing the states. Hence the corresponding post-measurement graphs satisfy $G_A[V_i] \le G_i \le \overline{G_B[V_i]}$ for $i = 1, 2$. Since $G_i \in \mathcal{G}$, there is an LOCC protocol to distinguish the remaining states. 

Finally, observe that $G_1 \vee G_2 = \overline{\left(\overline{G_1} \sqcup \overline{G_2}\right)}$ can be constructed from complements and disjoint unions, and hence the previous properties imply that $G_1 \vee G_2 \in \mathcal{G}$.
\end{proof}


Many natural classes of graphs are closed under taking of induced subgraphs. We would like to show that the set of distinguishable graphs $\mathcal{G}$ has this property. This is not entirely trivial. If $H$ is an induced subgraph of $G$, any representation of $G$ can be restricted to a representation of $H$, but we will generally not get every representation of $H$ this way. In order to prove this result, we will focus on local measurements that respect the structure of $G$, motivating the following definition.

\begin{definition}\label{G-preserving}
For a graph $G$ and a set of states $\{\ket{\psi_i}\}$ on $\mathcal H$ that is a representation of a spanning subgraph of $G$, we say that a measurement $\{E_k\}$ on $\mathcal H$ is {\bf $G$-preserving} if, for all $k$ and for all pairs of vertices  $v_i \not\sim v_j$ that are not neighbors in $G$, $\bra{\psi_i}E_k\ket{\psi_j} = 0$.
\end{definition}
We do not require the measurement to respect {\it all} orthogonal pairs, only those needed to remain a subgraph of $G$. 

\begin{lemma}\label{inducedsubgraphs}
    If $G \in \mathcal{G}$ and $H$ is an induced subgraph of $G$, then $H \in \mathcal{G}$. 
\end{lemma}

\begin{proof}
Suppose $G \in \mathcal{G}$ with $n$ vertices, and let $H$ be an induced subgraph of $G$ with $m<n$ vertices. We wish to show that $H \in \mathcal{G}$.

Let $\mathcal S = \{ \ket{\psi_k^A} \ot \ket{\psi_k^B} \}_{k = 1}^m$  be a set of product states with orthogonality graphs $H_A$ and  $H_B$  such that $H_A \le H \le \overline{H_B}$.  We extend these to spanning subgraphs of $G$ and $\overline{G}$ by adding $(n-m)$ isolated vertices to each of them. Define $G_A := H_A \sqcup \overline{K_{n-m}}$ and $G_B := H_B \sqcup \overline{K_{n-m}}$. 

Set $\mathcal{S}^\prime := \{\ket{\varphi_i^A} \ot \ket{\varphi_i^B} \}_{i=m+1}^{n}$, where $\{\ket{\varphi_i^A}\}$ and  $\{\ket{\varphi_i^B}\}$ are each orthonormal sets in extended Hilbert spaces such that for all $i,k$, $\bk{\varphi_i^A}{\psi_k^A} = \bk{\varphi_i^B}{\psi_k^B} = 0$. Then $G_A$ and $G_B$ are the Alice and Bob graphs for the set $\mathcal{S} \cup \mathcal{S}^\prime$.

The only edges in $G_A$ are the edges in $H_A$, so $G_A \le G$. Similarly, the only edges in $G_B$ are the edges in $H_B$, so  $H_B \le \overline{H}$ implies that $G_B \le \overline{G}$. Hence we have $G_A \le G \le \overline{G_B}$. As $G \in \mathcal{G}$, there exists an LOCC protocol that distinguishes the set of states  $\mathcal{S} \cup \mathcal{S}^\prime$, which implies that it distinguishes the original set of states $\mathcal S$, and thus $\mathcal{S}$ is LOCC-distinguishable. 

Therefore, for any set of product states $\mathcal S$ with orthogonality graphs $H_A$ and  $H_B$  such that $H_A \le H \le \overline{H_B}$, the states of $\mathcal{S}$ can be distinguished with LOCC, and this implies that $H \in \mathcal{G}$.
\end{proof}

While Lemma~\ref{inducedsubgraphs} started with a large graph $G \in \mathcal{G}$ and showed that its induced subgraphs remain in $\mathcal{G}$, it is frequently useful to be able to go in the other direction, to build up larger graphs from smaller ones.  In Example \ref{ex:not one way but three way}, Alice's graph has a cut vertex $v_5$, and the initial projective measurement in 
Figure~\ref{fig:locc-tree-example} produces two branches that are separated by this cut vertex.  This is an example of a more general phenomenon in which useful quantum measurements are suggested by separating sets of vertices that induce a clique. 

Going in the other direction, we can build a graph $G$ from smaller graphs $G_1$ and $G_2$ by identifying the vertices of a clique in $G_1$ with the vertices of a clique in $G_2$. This graph operation is called a {\it clique sum} (a notion that goes back almost a century \cite{wagner1937eigenschaft}), and is written $G = G_1 \oplus G_2$ when the cliques and the identifications are clear. 
Note that in some contexts, clique sums allow the removal of edges from the shared clique, but in what follows no such deletion is not permitted. For emphasis, we will frequently call it an edge-preserving clique sum.  

We establish the fact that $\mathcal{G}$ is closed under edge-preserving clique sums. As with induced subgraphs, any representation of $G_1, G_2, \ldots, G_r$ can be concatenated to create a representation of $G = \oplus G_i$, but not all representations of $G$ are obtained this way.

\begin{lemma}\label{cliquesum}
Let $G$ be a graph with cut set $V_0$ such that the induced graph $G[V_0]$ is a clique. Suppose that $G[V\backslash V_0]$ has $r \ge 2$ components. Define $V_i$ to be the vertices of the $i$th component and let $G_i = G[V_i \cup V_0]$ with $G_0 = G[V_0]$. 

If $G_1, G_2, \ldots, G_r \in  \mathcal{G}$, then $G \in  \mathcal{G}$. That is, $\mathcal{G}$ is closed under edge-preserving clique sums. 
\end{lemma}

\begin{proof}
Let $G$ be a graph satisfying the hypotheses of the lemma, and let $\mathcal S = \{ \ket{\psi_k^A} \ot \ket{\psi_k^B} \}_k$  be a set of product states on $\H_A \ot \H_B$ with associated graphs such that $G_A \le G \le \overline{G_B}$. 

Since $V_0$ is a cut set, we can decompose $\displaystyle \H_A = \oplus_{i = 0}^r \H_i$ into orthogonal subspaces such that for $i=1,\dots,r$, $\H_i$ is the span of $\{\ket{\psi_j^A}: v_j \in V_i\}$, and $\H_0$ is the (possibly trivial) orthogonal complement of the other subspaces. Alice can perform a measurement $\{P_i\}$ that projects onto the subspaces $\{\H_i\}$. Suppose that we get the outcome $k$. Then, our unknown state must correspond to a vertex in $V_k \cup V_0$ (even when $k=0$). All states corresponding to vertices outside this set are eliminated, leaving Alice with the reduced graph $G_A^{k}$. 

We show that this measurement is $G$-preserving by considering vertices $v_i, v_j$ that are not adjacent in $G$. This implies that $\bk{\psi_j^A}{\psi_i^A}=0$, since $G_A \le G$. Because $G[V_0]$ is a clique, at least one of these vertices (say $v_i$) is not in $V_0$. Since $\ket{\psi_i^A}$ is supported on $\H_i$, for each $k$, $P_k\ket{\psi_i^A} = \delta_{i,k} \ket{\psi_i^A}$. Hence $\bra{\psi_j^A} P_k \ket{\psi_i^A} = \delta_{i,k} \bk{\psi_j^A}{\psi_i^A} = 0$. So, the measurement is $G$-preserving, and Alice's post-measurement graph $G_A^{k}$ is a subgraph of $G_k$. 

We now show that we can complete the LOCC measurement. Let $U_k$ be the set of vertices $v_j$ for which $P_k \ket{\psi_j^A} \ne 0$. If $k \ne 0$, then $V_k \subseteq U_k \subseteq (V_k \cup V_0)$.
By Lemma~\ref{inducedsubgraphs}, $G_k \in \mathcal{G}$ implies that the induced subgraph $G_k[U_k]$ is also in $\mathcal{G}$. If we restrict Bob's graph to the remaining possible states, we write $G_B^{k} = G_B[U_k] \le \overline{G_k[U_k]}$. This implies that
\bee
G_A^{k} \le G_k[U_k] \le \overline{G_B^{k}} . 
\eee

The fact that $G_k[U_k] \in \mathcal{G}$ implies that any representation of $(G_A^{k}, G_B^{k})$ can be distinguished with LOCC, in particular our set $\mathcal{S}_k$. Hence, no matter which value of $k$ we get as a measurement outcome, we can complete the process of distinguishing the states $\mathcal S$ using LOCC, implying that $G \in \mathcal{G}$. 
\end{proof}

We construct an example in which we generate an LOCC measurement for a clique sum as described in Lemma~\ref{cliquesum}.

\begin{example}\label{cliquesum}
{\rm 
Let $H = \overline{P_5}$   be the house graph on 5 vertices. This is a 4-cycle attached to a 3-cycle, and it is the complement of a path on 5 vertices. Since its complement is chordal, $H \in \mathcal{G}$. 
 
Define the graph $G = H \oplus H$ to be the clique sum of two copies of $H$ joined along a 3-cycle such that the simplicial vertices are not identified with each other. In Figure~\ref{fig:clique-sum-example}, we note that the induced subgraphs on the first five vertices and on the last five vertices are each isomorphic to $H$.  

	\begin{figure}[h!]
		\begin{center}
        
\begin{tikzpicture}[scale=1.3, baseline=(current bounding box.north)]
\draw[very thick] (0,1)--(1,2)--(2,1)--(1,0)--cycle;
\draw[very thick]  (1,2)--(3,2);
\draw[very thick] (4,1)--(3,2)--(2,1)--(3,0)--cycle;
\draw[fill] (0,1) node[anchor=east] {$v_1$} circle [radius=2pt];
\draw[fill] (1,0) node[anchor=north] {$v_2$} circle [radius=2pt];
\draw[fill] (2,1) node[anchor=north] {$v_3$} circle [radius=2pt];
\draw[fill] (3,2) node[anchor=south] {$v_5$} circle [radius=2pt];
\draw[fill] (1,2) node[anchor=south] {$v_4$} circle [radius=2pt];
\draw[fill] (4,1) node[anchor=east] {$v_6$} circle [radius=2pt];
\draw[fill] (3,0) node[anchor=north] {$v_7$} circle [radius=2pt];
\end{tikzpicture} 

		\end{center}
			\caption{Graph clique sum $G = H \oplus H$ in Example~\ref{cliquesum}.}
            \label{fig:clique-sum-example}
	\end{figure}
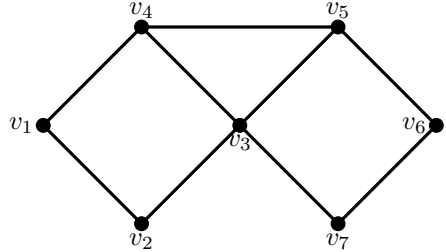

We claim that any product state representation with $G = G_A = \overline{G_B}$ of $G$ can be distinguished with LOCC by building on an existing protocol associated with the graph $H$. One such representation of (unnormalized) states is given below. Note that the specifics of Bob's representation are not needed to reduce the problem from $G$ to $H$. 
{\rm 
\begin{align*}
\ket{\psi_{1}}& = \ket{0}_A \ot \ket{1 + 2}_B &
 \ket{\psi_{2}} &=   \ket{0 + 1}_A\ot \ket{0+3}_B  \\
\ket{\psi_{3}} &= \ket{1+2}_A \ot \ket{2}_B & 
\ket{\psi_{4}} &=  \ket{0 - 1 + 4}_A\ot \ket{0}_B\\  
\ket{\psi_{5}} &= \ket{2 - 3 + 4 }_A\ot \ket{1}_B & 
\ket{\psi_{6}} &=  \ket{3}_A \ot \ket{2+3}_B \\ 
\ket{\psi_{7}} &=  \ket{2+3}_A\ot \ket{1+3}_B
\end{align*}}
We claim that any protocol to locally distinguish  states associated with $H$  can be incorporated into a protocol to distinguish all the states associated with $G$. 
We follow the procedure in the proof of Lemma~\ref{cliquesum}, setting $V_0 = \{v_3, v_4,v_5\}$, which induces a clique. Then $V_1 = \{ v_1, v_2 \}$ and $V_2 = \{v_6,v_7\}$. Alice performs the measurement $\{ P_1, P_2, P_0\}$, where $P_i$ projects onto the span of the states corresponding to $V_i$ for $i = 1,2$ and $P_0$ is everything else; that is,  
\bee 
P_1 = \kb{0}{0}_A + \kb{1}{1}_A, \qquad 
P_2 = \kb{2}{2}_A + \kb{3}{3}_A, \qquad 
P_0 = \kb{4}{4}_A .
\eee

If we get the outcome $P_0$, the remaining states (possible fewer than three) induce a clique in $G$ and form an independent set in $\overline{G} \ge G_B$, so Bob can distinguish them perfectly. 
 Suppose we get either outcome $P_1$ or $P_2$. In this case only four possible states remain, forming a copy of $C_4$ on $(v_1,v_2,v_3,v_4)$ or $(v_3,v_5,v_6,v_7)$. Since $C_4$ is an induced subgraph of $H$ and we assumed that $H \in \mathcal{G}$, we know that $C_4 \in \mathcal{G}$ from Lemma~\ref{inducedsubgraphs} and we can complete the measurement using any specified protocol for $H$, restricting Alice's first measurement to the span of the remaining states. 
}
\end{example}
 
By combining Lemma \ref{Lemma First Closure}, Lemma \ref{inducedsubgraphs}, and Lemma \ref{cliquesum}, we obtain the following result.

\begin{theorem}\label{ClassificationTheorem}
	The set ${\mathcal G}$ is closed under complements, disjoint unions, and joins. Further, it is closed under taking induced subgraphs and edge-preserving clique sums. 
\end{theorem} 

Theorem \ref{ClassificationTheorem} has immediate application in terms of where the set $\mathcal{G}$ of distinguishable graphs fits in terms of common graph classes. 

A \emph{cograph} \cite{corneil1981complement} is a graph that can be generated from an isolated vertex (that is, a vertex with no incident edges) by repeated use of graph complementation and disjoint union. Seven equivalent conditions for a graph to be a cograph are given by \cite[Theorem 2]{corneil1981complement}. The simplest and best known of these is that a graph is a cograph if and only if it has no induced subgraph that is isomorphic to $P_4$ (the path graph on 4 vertices). Since $K_1 \in \mathcal{G}$ and   $\mathcal{G}$ is closed under complements and disjoint unions, we have the following result.  

\begin{cor}\label{cographcor}
	The set of cographs is contained in the set $\mathcal{G}$.
\end{cor}

The set of cographs that are also split graphs is known as the set of threshold graphs; see, e.g.~\cite{mahadev1995threshold}. Threshold graphs can be constructed from a one-vertex graph through  repeated additions of a single vertex of two possible types: an isolated vertex   or a dominating vertex  (that is, a single vertex that is connected to all other vertices). In particular, the sets of split graphs and cographs are not subsets of one another, but both are subsets of $\mathcal{G}$ (by Theorem~\ref{splitchordalthm} and Corollary~\ref{cographcor}). 

Let us point out that $\mathcal G$ is not closed under every natural graph operation. Recall  that the tensor product $G_1\times G_2$ of graphs $G_1$ and $G_2$ is a graph having vertex set equal to  the cartesian product $V(G_1) \times  V(G_2)$, where vertices $(v_1,v_2)$ and $(v_1',v_2')$ are adjacent in $G_1\times G_2$  if and only if
$v_1$ is adjacent to $v_1'$ in $G_1$ and
$v_2$ is adjacent to $v_2'$ in $G_2$. If $G_1$ and $G_2$ have adjacency matrices $A_1$ and $A_2$, then the adjacency matrix of $G_1 \times G_2$ is $A_1 \ot A_2$. 
Similarly, the cartesian product $G_1\square G_2$ of graphs $G_1$ and $G_2$ is a graph having vertex set equal to  the cartesian product $V(G_1) \times  V(G_2)$, where vertices $(v_1,v_2)$ and $(v_1',v_2')$ are adjacent in $G_1\square G_2$  if and only if
$v_1$ is adjacent to $v_1'$ in $G_1$ and $v_2=v_2'$ , or $v_2$ is adjacent to $v_2'$ in $G_2$ and $v_1=v_1'$.

\begin{cor}
The set $\mathcal{G}$ is {\it not} closed under tensor products and is {\it not} closed under cartesian products. 
\end{cor}

\begin{proof} 
Setting $G_1 = K_2$ and $G_2 = K_3$, we see that $G_1, G_2 \in \mathcal{G}$. However, $K_2 \times K_3 = C_6$ and $K_2 \square K_3 = \overline{C_6}$, neither of which is a distinguishable graph by Proposition~\ref{unextendible} and the complement closure property in Lemma~\ref{Lemma First Closure}. Hence, $\mathcal{G}$ is {not} closed under tensor or cartesian products. 
\end{proof}

Finally, we show that the set of distinguishable graphs is properly contained inside a particular set of graphs. 

\begin{definition}
    A graph $G$ is {\bf weakly chordal} if neither $G$ nor its complement $\overline{G}$ contains an induced cycle $C_n$ with $n \ge 5$. 
\end{definition}

Weakly chordal graphs were first defined and studied in \cite{hayward1985weakly}.

\begin{cor}
    The set of distinguishable graphs $\mathcal{G}$ is a subset of the weakly chordal graphs, and this inclusion is proper. 
\end{cor}

\begin{proof}
Fix $n \ge 5$. From Proposition \ref{unextendible}, we know that $C_n \notin \mathcal{G}$. Since $\mathcal{G}$ is closed under induced subgraphs and complements, $G  \in \mathcal{G}$ implies that neither $G$ or $\overline{G}$ contains an induced cycle $C_n$. 
Since this is true for all $n\ge 5$, $G$ is weakly chordal.

On the other hand, the states in Example \ref{ex:not one way but three way} are a subset of the nonlocality without entanglement states in \cite{bennett1999quantum}, where it is shown that the first eight states cannot be perfectly distinguished with any finite number of rounds of LOCC, hence the graph representing them is not in $\mathcal{G}$. This graph $G$ is given in Figure \ref{fig:NLWE Graph} and is easily seen to be weakly chordal despite not being in $\mathcal{G}$. Hence the inclusion is proper. 
	\begin{figure}[h!]
		\begin{center}
			\begin{tikzpicture}[scale=1.15, baseline=(current bounding box.north)]
			
			
			\draw[fill] (-2,1)   node[anchor=south] {$1$}     circle (2pt);
			\draw[fill] (-2,0)   node[anchor=north] {$2$}     circle (2pt);
			
			\draw[fill] (-0.7,1) node[anchor=south] {$6$}   circle (2pt);
			\draw[fill] (-0.7,0) node[anchor=north] {$7$}   circle (2pt);
			
			\draw[fill] (0.7,0)  node[anchor=north] {$5$}   circle (2pt);

            \draw[fill] (0.7,1)  node[anchor=south] {$8$}   circle (2pt);

			\draw[fill] (2,0)    node[anchor=north] {$3$}     circle (2pt);
			\draw[fill] (2,1)    node[anchor=south] {$4$}   circle (2pt);
			
			\draw[thick] (-2,1)--(-2,0);
			\draw[thick] (-2,1)--(-0.7,1);
			\draw[thick] (-2,1)--(-0.7,0);
			
			\draw[thick] (-2,0)--(-0.7,0);
			\draw[thick] (-2,0)--(-0.7,1);
			
			\draw[thick] (-0.7,1)--(0.7,0);
			
			\draw[thick] (-0.7,0)--(0.7,0);
			\draw[thick] (0.7,1)--(-0.7,0);
			
			\draw[thick] (0.7,1)--(2,0);
			\draw[thick] (0.7,0)--(2,1);
			\draw[thick] (0.7,0)--(2,0);
			\draw[thick] (-2,1)--(2,1)--(2,0);
			
			\end{tikzpicture}
		\end{center}
			\caption{$G = G_A = G_B \notin \mathcal{G}$}
            \label{fig:NLWE Graph}
	\end{figure}
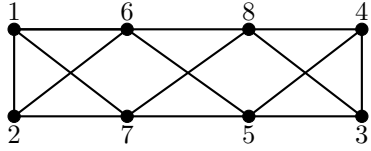
\end{proof}

   We indicate the known graph class relationships in Figure \ref{fig:GraphInclusions}.

	\begin{figure}[h!]
		\begin{center}
\scalebox{0.5}{\begin{tikzpicture}

    \draw[draw = black] (-1.2,0) circle (1.4);
      \draw[draw = black] (-1.2,0) circle (2.3);
     \draw[draw = black] (1.2,0) circle (2.3);
     \draw[draw = black] (0,1) circle (5);
      \draw[draw = black] (0,0.2) circle (4);
  \node at (0,5) {\large\textbf{Weakly Chordal}};
\node at (-3,0) {\LARGE\textbf{$\mathcal{G}_1$}};
    \node at (0,3.4) {\LARGE\textbf{$\mathcal{G}$}};
    \node at (-1.8,0) {\LARGE\textbf{$\mathcal{G}_0$}};
        \node at (2.1,-.2) {\large\textbf{Cographs}};
\end{tikzpicture}}
\end{center}
			\caption{Graph classes relative to $\mathcal{G}$. $\mathcal{G}_1$ is the set of chordal graphs, and $\mathcal{G}_0$ is the set of split graphs. }
\label{fig:GraphInclusions}
\end{figure}
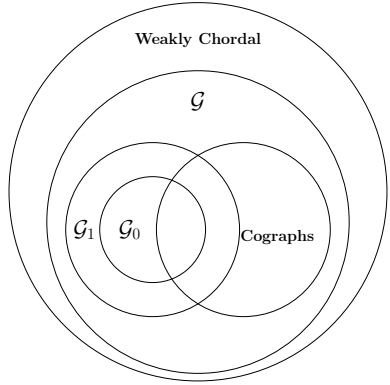

\section{Conclusion and Outlook}

This work opens up several new lines of investigation. 
Most obvious is a direct continuation of the work we have begun on determining the structure of the class of distinguishable graphs for full LOCC. The current work is largely concerned with determining and deriving the basic graph closure properties of the set, along with relating the set of distinguishable graphs to other well-known graph classes. Natural next steps in this direction would be to conduct a deeper analysis on the distinguishability properties of specific subclasses of graphs, including those that arise in LOCC communication applications.   
Moreover, the results in \cite{KMNP2024} on one-way LOCC established links between local quantum measurements, graph theory, and the theory of matrix completions, and all of these directions can now be pursued in depth for full LOCC by building on this work.

In terms of matrix completions, we note that the set of representations of a graph corresponds via Gram matrices contained in the cone  of positive semidefinite matrices with a given zero pattern, matching that of the graph's adjacency matrix. Thus, every result about classes of graphs has a corresponding result in terms of adjacency matrices and their associated cones. We showed that $\mathcal{G}$ was closed under complements and disjoint unions. Recall the fact that $A(\overline{G})=J-I-A(G)$, where $A(G)$ is the adjacency matrix of $G$ and $J$ is the all 1's matrix, and  $A(G_1\sqcup G_2)=A(G_1)\oplus A(G_2)$. One can see how the roles of $G_A$ and $G_B$ are interchanged in $A(\overline{G})$ since all the off-diagonal entries get ``flipped''---0's become 1's and 1's become 0's. One can also see that the adjacency matrix  $A(G_1\sqcup G_2)$ is block-diagonal, leading to the orthogonal subspaces argument used herein. The final two operations simply reduce the adjacency matrix of the original graph; such reductions do not change the LOCC-distinguishability of the graph, as we saw in Section~\ref{sec:structure}. Similarly, to construct the adjacency matrix of an induced subgraph,  we simply delete the appropriate rows/columns: $A(G')=A(G)_{(-i,-i)}$, and the  adjacency matrix of an edge-preserving clique sum is constructed by combining  $A(G_1)$ and $A(G_2)$ by merging the rows and columns that correspond to the identified pairs of vertices in the shared $k$-clique. We showed that $\mathcal G$ is closed under all of these operations. 

We also showed that $\mathcal{G}$ is {\it not} closed under tensor and cartesian products. Looking at the adjacency matrices, we see that  $A(G_1\times G_2)=A(G_1)\otimes A(G_2)$ and $A(G_1\square G_2)=A(G_1)\otimes I+I\otimes A(G_2)$, operations which distort the adjacency matrices in a manner that does not preserve LOCC-distinguishability.  Another possible graph operation to explore is the $\ltimes$ product  defined in \cite{coutinho2016perfect} as a generalization of the double cover  operation, for which graphs $G_1$ and $G_2$ on the same vertex set generate a new graph with adjacency matrix $A(G_1 
\ltimes G_2)=\begin{pmatrix}
    A(G_1)& A(G_2)\\ A(G_2)& A(G_1)
\end{pmatrix}$. A characterization of when $\mathcal G$ is closed or not in terms of the adjacency matrices would prove useful, especially in cases when the graph operation, e.g., $\ltimes$, is defined via its adjacency matrix. 

On another topic, when we move the conversation from full quantum state discrimination to LOCC, deciding whether states are distinguishable is equivalent to finding a sequence of local measurements that eliminate some states but do not disturb the remaining states too much. In the case of one-way LOCC, in general each local measurement must eliminate the complement of a clique, and for one-way bipartite LOCC we know this is always possible if and only if the corresponding graph is chordal. This follows from the general result that any set of quantum states whose orthogonality graph is chordal can be reduced to an induced clique with a quantum measurement. 
As we move to two-way LOCC, we need to consider a sequence of nontrivial steps that accomplish discrimination. Each of these steps should, in general, be an orthogonality-preserving quantum measurement, and we can require that it be a nontrivial measurement so that our set of possible states post-measurement is strictly smaller than what we started with. This description is essentially the same whether we are talking about bipartite or multipartite discrimination. It only depends on the collection of local graphs. Thus, it would be natural to pursue an extension of the approach introduced here to multipartite LOCC.  
One can easily generalize from a two-party system $\mathcal H=\mathcal H_A\otimes \mathcal H_B$ to a $d$-party system $\mathcal H=\bigotimes_{i=1}^d \mathcal H_i$. A multipartite product state would then be of the form $\ket{\psi}=\ket{\psi_1}\otimes \ket{\psi_2}\otimes \cdots \otimes \ket{\psi_d}$. 
If we have $d$ parties and a set of multipartite product states, let $G_i = (V,E_i)$ be the local orthogonality graph associated with the $i$th party. 
The states are orthogonal if and only if $\bigcap_{i = 1}^d E_i$ is empty. 
The interplay between these graphs and the implications for distinguishability can then be considered. 
LOCC in a multipartite context is frequently quite challenging, but we think it should be possible to extend at least some of the graph theory approach to the problem of locally distinguishing sets of multipartite product states. 

One could also consider this work in the context of general quantum state discrimination, which is equivalent to asking how much one local party can discover without communication from anyone else. 
The problem of distinguishing among a set of quantum states is one of the oldest problems in quantum information theory. A set of pure states can be distinguished perfectly if and only if they are mutually orthogonal; but outside of this special circumstance, there are many possible ways to define optimal discrimination. The results in \cite{KMNP2024} incorporated the problem of identifying subsets of pure quantum states, but other paradigms exist such as minimum error, unambiguous, antidiscrimination, and orthogonality-preserving discrimination. Although some of these conditions have been demonstrated using graphs, there does not appear to have been attempts to classify graphs that always work as we have considered in this work.

We plan to undertake some of these investigations elsewhere, and we invite other interested researchers to do the same.

\vspace{0.1in} 

\noindent{\bf Acknowledgements.}
 D.W.K.\ was supported by NSERC Discovery Grant number 400160.  R.P. \ was supported by NSERC Discovery Grant number 400550. S.P.\ was supported by NSERC Discovery Grant number 1174582, the Canada Foundation for Innovation (CFI) grant number 35711, and the Canada Research Chairs (CRC) Program grant number 231250.

\bibliographystyle{plain}
\bibliography{references}

\end{document}